\documentclass[11pt]{article}

\usepackage{amsmath,amssymb,amsfonts,bbm,bm}
\usepackage{slashed}

\usepackage{tocloft} 

\usepackage[nosort]{cite} 
\usepackage{graphicx,color}
\usepackage[colorlinks=true, linkcolor=blue, citecolor=blue, linktoc=all]{hyperref}

\usepackage[usenames,dvipsnames]{xcolor}

\usepackage{tikz-cd}
\usepackage{pict2e}
\usepackage{physics}
\usepackage{comment} 
\usepackage[makeroom]{cancel}
\usepackage{musicography}

\newcommand{\beq}{\begin{eqnarray}}
\newcommand{\eeq}{\end{eqnarray}}
\newcommand{\beqn}{\begin{eqnarray}}
\newcommand{\eeqn}{\end{eqnarray}}
\newcommand{\bea}{\begin{eqnarray}}
\newcommand{\eea}{\end{eqnarray}}
\newcommand{\be}{\begin{equation}}
\newcommand{\ee}{\end{equation}}

\newcommand{\un}[1]{\underline{#1}}

\usepackage[normalem]{ulem}

\newcommand{\uiuc}[1]{
	\centerline{
		\begin{minipage}[c]{0.7\textwidth}
			\begin{center}
			${}^{#1}$ Illinois Center for Advanced Studies of the Universe \& Department of Physics,\\ 
			University of Illinois, 1110 West Green St., Urbana IL 61801, U.S.A.
			\end{center}
		\end{minipage}
		}
	}
 \newcommand{\nyu}[1]{
	\centerline{
		\begin{minipage}[c]{0.7\textwidth}
			\begin{center}
			${}^{#1}$ Center for Cosmology and Particle Physics, New York University, New York, NY 10003, USA
			\end{center}
		\end{minipage}
		}
	}
  \newcommand{\upenn}[1]{
	\centerline{
		\begin{minipage}[c]{0.7\textwidth}
			\begin{center}
			${}^{#1}$David Rittenhouse Laboratory, University of Pennsylvania, Philadelphia, PA 19104, USA
			\end{center}
		\end{minipage}
		}
	}

\usepackage{mathrsfs,mathtools}
\renewcommand\mathbb[1]{\mathbbm{#1}}

\newcommand\saa[1]{{\color{teal}[\textbf{Shadi}: #1]}}

\makeatletter
\DeclareRobustCommand{\loplus}{\mathbin{\mathpalette\dog@lsemi{+}}}
\DeclareRobustCommand{\lotimes}{\mathbin{\mathpalette\dog@lsemi{\times}}}
\DeclareRobustCommand{\roplus}{\mathbin{\mathpalette\dog@rsemi{+}}}
\DeclareRobustCommand{\rotimes}{\mathbin{\mathpalette\dog@rsemi{\times}}}

\newcommand{\dog@rsemi}[2]{\dog@semi{#1}{#2}{-90,90}}
\newcommand{\dog@lsemi}[2]{\dog@semi{#1}{#2}{270,90}}
\newcommand{\dog@semi}[3]{%
  \begingroup
  \sbox\z@{$\m@th#1#2$}%
  \setlength{\unitlength}{\dimexpr\ht\z@+\dp\z@\relax}%
  \makebox[\wd\z@]{\raisebox{-\dp\z@}{%
    \begin{picture}(1,1)
    \linethickness{\variable@rule{#1}}
    \roundcap
    \put(0.5,0.5){\makebox(0,0){\raisebox{\dp\z@}{$\m@th#1#2$}}}
    \put(0.5,0.5){\arc[#3]{0.5}}
    \end{picture}%
  }}%
  \endgroup
}
\newcommand{\variable@rule}[1]{%
  \fontdimen8  
  \ifx#1\displaystyle\textfont3\else
    \ifx#1\textstyle\textfont3\else
      \ifx#1\scriptstyle\scriptfont3\else
        \scriptscriptfont3\relax
  \fi\fi\fi
}
\DeclareRobustCommand{\loplus}{\mathbin{\mathpalette\dog@lsemi{+}}}

\newcommand{\nn}{\nonumber}
\usepackage{float}
\usepackage{musicography}

\usepackage{tikz-cd}
\newcommand{\thistitle}{Relational Quantum Geometry}

\begin{document}

\title{\thistitle}
\author{
Shadi Ali Ahmad$^{a},$ Wissam Chemissany$^{b},$ Marc S. Klinger$^{c}$, and Robert G. Leigh$^{c}$
	\\
	\\
    {\small \emph{\nyu{a}}} \\ \\
    {\small \emph{\upenn{b}}} \\ \\
	{\small \emph{\uiuc{c}}}
	\\
	}
\date{}
\maketitle
\vspace{-0.5cm}
\begin{abstract}
A common feature of the extended phase space of gauge theory, the crossed product of quantum theory, and quantum reference frames (QRFs) is the adjoining of degrees of freedom followed by a constraining procedure for the resulting total system. Building on previous work, we identify non-commutative or quantum geometry as a mathematical framework which unifies these three objects. We first provide a rigorous account of the extended phase space, and demonstrate that it can be regarded as a classical principal bundle with a Poisson manifold base. We then show that the crossed product is a trivial quantum principal bundle which both substantiates a conjecture on the quantization of the extended phase space and facilitates a relational interpretation. Combining several crossed products with possibly distinct structure groups into a single object, we arrive at a novel definition of a quantum orbifold. We demonstrate that change of frame maps within the quantum orbifold correspond to quantum gauge transformations, which are QRF preserving maps between crossed product algebras. Finally, we conclude that the quantum orbifold is equivalent to the G-framed algebra proposed in prior work, thereby placing systems containing multiple QRFs squarely in the context of quantum geometry.
\vspace{0.3cm}
\end{abstract}
\begingroup
\hypersetup{linkcolor=black}
\tableofcontents
\endgroup
\noindent\rule{\textwidth}{0.6pt}
\setcounter{footnote}{0}
\renewcommand{\thefootnote}{\arabic{footnote}}

\newcommand{\curr}[1]{\mathbb{J}_{#1}}
\newcommand{\constr}[1]{\mathbb{M}_{#1}}
\newcommand{\chgdens}[1]{\mathbb{Q}_{#1}}
\newcommand{\spac}[1]{S_{#1}}
\newcommand{\hyper}[1]{\Sigma_{#1}}
\newcommand{\chg}[1]{\mathbb{H}_{#1}}
\newcommand{\ThomSig}[1]{\hat{\Sigma}_{#1}}
\newcommand{\ThomS}[1]{\hat{S}_{#1}}
\newcommand{\discuss}[1]{{\color{red} #1}}

\newcommand{\gravgroup}{G_L}
\newcommand{\gaugegroup}{G_G}
\newcommand{\gravL}{L_L}
\newcommand{\gaugeL}{L_G} 
\newcommand{\solder}{e}
\newcommand{\sym}{\Omega}
\newcommand{\symdens}{\omega}
\newcommand{\sympot}{\Theta}
\newcommand{\sympotdens}{\theta}
\newcommand{\chgcon}[1]{\chg{#1}^{(1)}}
\newcommand{\chgchg}[1]{\chg{#1}^{(2)}}
\newcommand{\cpalg}{\mathcal{M}_{cp}}

\newtheorem{theorem}{Theorem}[section]
\newtheorem{example}{Example}[section]
\newtheorem{proof}{Proof}[section]
\newtheorem{definition}{Definition}[section]

\newcommand{\red}[1]{{\color{red} #1}}
\newcommand{\purple}[1]{{\color{purple} #1}}

\section{Introduction}
The problem of constrained quantum systems and gauge theories is a very intricate one with a long history~\cite{dirac1964lectures}. The starting point is a system admitting a gauge symmetry captured by a group $G$. That the symmetry is gauged means that observables related by group transformations must be identified rather than treated as physically distinct. By consequence, the constrained physics is encoded in the subset of observables which are invariant under the group action. With this being said, the procedure of passing from the full set of kinematical observables to its gauge-invariant subset is not without its technical challenges. Addressing one of the more thorny subtleties of this approach, it is sometimes the case that the constraining procedure completely trivializes the set of legal observables~\cite{Higuchi:1991tk, Higuchi:1991tm, Marolf:2008hg, chandrasekaran_algebra_2022}. Does that mean there are \textit{no} degrees of freedom in the theory? Surely, this cannot be the case in full generality. A fruitful way forward is to adjoin additional degrees of freedom carrying a representation of the symmetry group to act as a reference frame for the overall system. Together, the original system and the reference frame degrees of freedom must now be constrained appropriately to reduce to a gauge invariant theory. It has been observed, and was in fact proven in \cite{ahmad2024quantumreferenceframestopdown}, that this procedure always produces a non-empty set of physical observables. 

The procedure of adjoining degrees of freedom to `complete' a physical theory has three significant avatars in different fields of physics:
\begin{enumerate}
    \item \textbf{Extended phase space:} At the classical level, a system is described by a phase space. In the case of a constrained system, one needs to perform a reduction from the full phase space to a restricted version on which the constraints of the theory are implemented~\cite{henneaux_quantization_1992}. Mathematically, to perform such a reduction one must ensure that the symmetry action is realized equivariantly on the phase space. In physical parlance, this is related to the integrability of the Noether charges of the theory. In cases where the symmetry charges are \textit{not} integrable~\cite{marsden1974reduction}, one may first move to a bigger phase space, termed the extended phase space, to ensure the existence of a proper reduction. In ~\cite{Ciambelli:2021nmv, Klinger:2023qna, Jia:2023tki} it has been demonstrated how, given a phase space with a symmetry action as captured by a \textit{Poisson dynamical system}, one can always construct such an extended phase space in which the symmetry charges are integrable. In field theory scenarios, the additional degrees of freedom which are added into the theory often encode what are conventionally referred to as `edge modes' \cite{Donnelly:2022kfs,Freidel:2021dxw,Freidel:2021cjp,Donnelly:2020xgu,Freidel:2020xyx,Freidel:2020svx,Freidel:2020ayo,Carrozza:2022xut,Carrozza:2021gju,Goeller:2022rsx}, and are intimately related to the corner proposal \cite{Ciambelli:2021vnn,Ciambelli:2022cfr}. 
    \item \textbf{Crossed product algebras:} At the quantum level, a system is described by an algebra of observables which is generically $C^{*}$ to ensure a Hilbert space representation, and which therefore may be further completed to a von Neumann algebra~\cite{Takesaki1, Takesaki2, Takesaki3}. If the system admits a symmetry in the form of a group automorphism action we can ask whether the symmetry is implemented via conjugation by a representation of the group inside of the algebra. When this is the case we call the automorphism action inner, else it is termed outer. This is an algebraic analog of the integrability of charges in phase space. The collection of an algebra together with a group automorphism action is called an \emph{algebraic dynamical system}. Given any such system one can construct an extended algebra called the \emph{crossed product} which is obtained by combining degrees of freedom from the original algebra with a unitary representation of the group acting compatibly on a single enlarged Hilbert space. The crossed product can alternatively be thought of as a tensor product of system and symmetry generators \textit{constrained} by an extended symmetry action. By construction, the automorphism action used to define a crossed product is manifestly inner therein, positioning the crossed product as an algebraic foil to the extended phase space. 
    \item \textbf{Quantum reference frames}: A famous problem in quantum gravity is the problem of time~\cite{schopf_battelle_1970, isham_canonical_1993,Anderson:2010xm,rovelli_time_1991, Kuchar1991-KUCTPO-4, barbour_timelessness_1994,barbour_timelessness2_1994,butterfield_arguments_2006, kiefer_quantum_2012, rovelli_quantum_2004}, which arises because of the Hamiltonian constraint of general relativity. Intuitively, this constraint is a consequence of the inability to pick a preferred time in a metric theory of gravity. There are several well documented problems which arise in constraining gravity to respect the Hamiltonian constraint \cite{Ashtekar1991-ASHCPO}. One potential resolution is relational: adjoin additional degrees of freedom which can serve as a reference frame, or clock, for the gravitational degrees of freedom~\cite{Domagala:2010bm, Thiemann:2006up, Husain:2011tk}. In toy models with a Hamiltonian constraint~\cite{page_evolution_1983,wootters_time_1984} one adjoins a clock algebra of observables to the system and constrains the total system. Given a more general symmetry group, one adjoins an appropriate reference algebra which is a generalization of clock observables. This procedure instantiates the quantum reference frame (QRF) program~\cite{aharonov_quantum_1984, rovelli_quantum_1991, kitaev_superselection_2004, bartlett_reference_2007, gour_resource_2008,  girelli_quantum_2008, bartlett_quantum_2009, angelo_physics_2011, angelo_kinematics_2012,  palmer_changing_2014, pienaar_relational_2016, smith_quantum_2016, miyadera_approximating_2016, loveridge_relativity_2017, belenchia_quantum_2018, giacomini_quantum_2019, hoehn_equivalence_2020, castro-ruiz_quantum_2020, Danielson:2021egj,hoehn_trinity_2021, castro-ruiz_relative_2021, giacomini_spacetime_2021, ali_ahmad_quantum_2022, cepollaro_quantum_2022, giacomini_second-quantized_2022, apadula_quantum_2022}, which treats reference frames as quantum systems in their own right as dictated by operational principles.
\end{enumerate}

In previous work, we have proposed useful correspondences between these three arenas. In Refs.~\cite{Klinger:2023auu, Klinger:2023tgi}, it was conjectured that the quantization of an extended phase space is a crossed product algebra in which the algebraic dynamical system of the latter is functorially related to the Poisson dynamical system of the former. In Ref.~\cite{ahmad2024quantumreferenceframestopdown} it was illustrated that the crossed product can be interpreted directly as a quantum reference frame. In Ref.~\cite{Ciambelli:2024swv}, it has been argued that the relational resolution to the problem of time comes to pass in the quantization of gravity on null hypersurfaces via a careful implementation of the extended phase space.   In this work, we unify these three fields and realize them as different limits of a single mathematical structure called the quantum principal bundle (QPB)~\cite{Brzezinski:1992hda, budzynski1994quantumprincipalfiberbundles, Calow_2002, beggs_quantum_2020 }, or, to be more precise, a slightly more general structure we introduce called a quantum orbifold. 

We begin in Section \ref{sec: extended phase space} by providing a rigorous mathematical definition of the extended phase space, and demonstrating that it possesses the structure of a classical principal bundle whose base manifold is Poisson. We also recall the definition of a von Neumann crossed product algebra starting from a covariant representation, and exhibit the compelling correspondence between the crossed product and the extended phase space. In Section \ref{sec: Intro QPB} we introduce the notion of a trivial quantum principal bundle by quantizing the classical definition in a sequence of two steps. First, we reformulate the standard definition in terms of commutative algebras of functions on measure spaces. Second, we promote the function algebras to general non-commutative algebras. This results in a notion of quantum principal bundle which automatically reproduces the classical version in the appropriate limit. In Section \ref{sec: CP = QPB} we provide an explicit proof equating \textit{trivial} QPBs satisfying certain analytic conditions and von Neumann crossed product algebras. This result settles the conjecture of~\cite{Klinger:2023auu, Klinger:2023tgi} identifying the quantization of the extended phase space with the crossed product. Utilizing the insights of~\cite{ahmad2024quantumreferenceframestopdown}, in Section \ref{sec: Relational QPBs} we provide a relational interpretation for QPBs. Most crucially, in Section \ref{sec: QGT} we introduce quantum gauge transformations and demonstrate that they encode QRF preserving maps between crossed product algebras. We exhibit the role of such maps in the context of constraint quantization in Section \ref{sec: constraint quant}. Finally, in Section \ref{sec: locally trivial} we provide an algebraic description of \emph{locally} trivial quantum principal bundles. We conclude by realizing our previously defined G-framed algebra, a structure housing multiple different crossed products and so different QRFs, as a \textit{quantum orbifold}. 

For the benefit of the reader, we summarize the logical progression of our paper in Figure \ref{fig: graphic}. We hope that this note leads to a productive synergy among the high energy theory, quantum gravity, QRF, and quantum geometry communities. As a side goal, our recent work aims to bring the attention of the high energy community to the crossed product \textit{independently} of its usefulness in obtaining semifinite algebras from which entropic observables can be computed~\cite{Klinger:2023auu,witten_gravity_2022, chandrasekaran_large_2022, Sorce:2023fdx, AliAhmad:2023etg, AliAhmad:2024eun}.

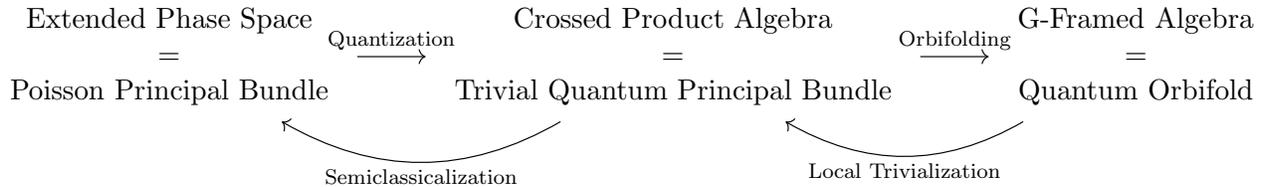
\begin{figure}
    \centering
    \begin{tikzcd}
    \text{\begin{tabular}{c}
        Extended Phase Space \\
        $=$ \\
        Poisson Principal Bundle
    \end{tabular}}
        \arrow[r, "\text{Quantization}"] &
        \arrow[bend left]{l}{\text{Semiclassicalization}}
    \text{\begin{tabular}{c}
        Crossed Product Algebra \\
        $=$ \\
        Trivial Quantum Principal Bundle
    \end{tabular}}
        \arrow[r, "\text{Orbifolding}"] & 
        \arrow[bend left]{l}{\text{Local Trivialization}}
    \text{\begin{tabular}{c}
        G-Framed Algebra \\
        $=$ \\
        Quantum Orbifold
    \end{tabular}}
    \end{tikzcd}
    \caption{Graphic depiction of our results.}
    \label{fig: graphic}
\end{figure}

\section{Crossed Products as Quantum Principal Bundles}

In this section we recall the notion of a quantum principal bundle as introduced in \cite{ahmad2024quantumreferenceframestopdown}. We begin by reviewing the definition of a principal bundle in the classical context and by providing a useful example in the form of the extended phase space of a physical system acted upon by a Lie group. We then demonstrate how the standard definition of a principal bundle can be translated from statements about topological spaces to statements about their algebras of functions. Promoting this translation to the status of a definition and relaxing the assumption that the algebras involved be function algebras for underlying spaces, we arrive at a fully non-commutative definition of a principal bundle. We conclude by proving a one to one correspondence between quantum principal bundles satisfying certain analytic conditions and crossed products of von Neumann algebras by locally compact groups. This proof provides a non-commutative geometric point of view on the correspondence proposed in \cite{Klinger:2023auu, Klinger:2023tgi} identifying the extended phase space on the classical side and the crossed product on the quantum side. Our aim here is two-fold: bringing in the machinery of non-commutative geometry to the physics community and promoting the fundamental status of the crossed product by viewing it as a structure as indispensable to quantum theory as phase space is to classical theory.

\subsection{Classical Principal Bundles and the Extended Phase Space} \label{sec: extended phase space} 

A principal bundle is a triple $(P,B,G)$ consisting of two manifolds, $P$ and $B$, and a Lie group, $G$. The manifold $P$ is called the total space of the principal bundle and admits a pair of structure maps. The first is a projection $\pi: P \rightarrow B$ which identifies the manifold $B$ as the base space of $P$ when regarded as a bundle. The second is a free right action $R: G \times P \rightarrow P$ which identifies $G$ in its role as structure group. Recall that $R$ is free if
\beq
	R_g(p) = p \implies g = e,
\eeq
where $e \in G$ is the identity of the group. In this section we restrict our attention to locally trivial principal bundles. This means that, for every open subset $U \subset B$, there exist a map $t_U: \pi^{-1}(U) \rightarrow G$ such that
\beq \label{trivialization condition}
	t_U \circ R_g(p) = t_U(p) g. 
\eeq
Of course, this also implies that the full principal bundle can be trivialized in each open chart via the map $\phi_{U}: \pi^{-1}(U) \rightarrow U \times G$, with $\phi_U(p) = (\pi(p), t_U(p)), \; \forall p \in \pi^{-1}(U)$. 

For our purposes, a good example of a principal bundle comes from the notion of the extended phase space. Let $(X,\Pi)$ be a Poisson manifold associated with the phase space of a physical system.\footnote{In many cases we can interpret $X = T^*Q$ as the cotangent bundle of a configuration space $Q$.} The map $\Pi: \bigwedge^{2} T^*X \rightarrow C^{\infty}(X)$ is called a Poisson tensor and satisfies the equation
\beq
	[\Pi,\Pi]_{S} = 0,
\eeq
with $[,]_{S}$ the Schouten-Nijenhuis bracket.\footnote{The Schouten-Nijenhuis generalizes the Lie bracket to multi-vector fields. We define this bracket by its action on wedge products of vectors. Let $\{\un{X}_i\}_{i = 1}^k$ and $\{\un{Y}_j\}_{j = 1}^l$ be a pair of sets of vector fields on $X$. Then
\beq
	[\un{X}_1 \wedge ... \wedge \un{X}_k, \un{Y}_1 \wedge ... \wedge \un{Y}_l]_S \equiv \sum_{i,j} (-1)^{i+j} [\un{X}_i,\un{Y}_j] \wedge \un{X}_1 \wedge ... \wedge \widehat{\un{X}_i} \wedge ... \wedge \un{X}_k \wedge \un{Y}_1 \wedge ... \wedge \widehat{\un{Y}_j} \wedge ... \wedge \un{Y}_l
\eeq}
The Poisson tensor defines a Poisson bracket on the algebra of functions $C^{\infty}(X)$ by the equation
\beq
	\{f,g\}_{\Pi} \equiv \Pi(df,dg). 
\eeq

We say that a group $G$ acts as a symmetry on $X$ if it realizes a homomorphism $a: G \rightarrow \text{Aut}(X)$. Here, $\text{Aut}(X)$ is the set of diffeomorphisms of $X$ that preserve the Poisson tensor, e.g. the set of Poisson maps from $X$ to itself. Restricting our attention to Lie groups, we can study the infinitesimal structure of the symmetry action by passing from the group to its associated Lie algebra $\mathfrak{g}$. Then, we realize a map $\xi: \mathfrak{g} \rightarrow TX$ such that
\beq \label{Infinitesimal symmetry}
	a_{\text{exp}(\un{\mu})} = \text{exp}(\xi_{\un{\mu}}), \; \forall \un{\mu} \in \mathfrak{g}.
\eeq
On the left hand side of \eqref{Infinitesimal symmetry} $\text{exp}: \mathfrak{g} \rightarrow G$ is the exponentiation of the Lie algebra to the (connected component of) the group $G$ whereas on the right hand side $\text{exp}: TX \rightarrow \text{Diff}(X)$ constructs the integral curve of a tangent vector. Since $a_g$ is a Poisson map, the associated vector field $\xi_{\un{\mu}}$ will be a Poisson vector field, meaning $\mathcal{L}_{\xi_{\un{\mu}}} \Pi = 0$ for each $\un{\mu} \in \mathfrak{g}$. 

We refer to the triple $(X,G,a)$ as a Poisson covariant system. To each such triple we can identify an enlarged Poisson manifold $X \times_a G$ which we refer to as the extended phase space \cite{Ciambelli:2021nmv, Klinger:2023qna}. The extended phase space is a bundle over $X$ based locally around the product manifold $X \times \mathfrak{g}^*$. In other words, it admits a projection $\pi: X \times_a G \rightarrow X$, and for each $U \in X$ there exists a trivialization $\phi_U: \pi^{-1}(U) \rightarrow X \times \mathfrak{g}^*$. For simplicity, we work from here on purely within such a local trivialization. 

The Poisson tensor of $X \times_a G$ is inherited from (i) the Poisson tensor on $X$, (ii) the canonical Poisson tensor on $\mathfrak{g}^*$, and (iii) the infinitesimal action of $G$ on $X$. Stated as a Poisson bracket we have, for $\mathfrak{F},\mathfrak{G} \in C^{\infty}(X \times \mathfrak{g}^*)$:
\beq \label{extended bracket}
	\{\mathfrak{F},\mathfrak{G}\}_{ext}(\alpha,x) \equiv \{\mathfrak{F}(\alpha,\cdot), \mathfrak{G}(\alpha,\cdot)\}_X(x) + \{\mathfrak{F}(\cdot,x),\mathfrak{G}(\cdot,x)\}_{\mathfrak{g}^*} - \bigg(\mathcal{L}_{\xi_{\lambda(\delta \mathfrak{F}(\alpha,x))}} \mathfrak{G}(\alpha,x) - \mathcal{L}_{\xi_{\lambda(\delta \mathfrak{G}(\alpha,x))}} \mathfrak{F}(\alpha,x)\bigg). 
\eeq
The Lie derivative appearing in \eqref{extended bracket} is that of $X$, while the symbol $\delta$ denotes the exterior derivative on $\mathfrak{g}^*$. The map $\lambda: \Omega^1(\mathfrak{g}^*) \rightarrow \mathfrak{g}$ is an isomorphism defined by recognizing that the set of one forms on $\mathfrak{g}^*$ are linear functionals on $\mathfrak{g}^*$, and hence are isomorphic to $\mathfrak{g}^{**} \simeq \mathfrak{g}$. With this map in hand, the Poisson bracket on $\mathfrak{g}^*$ is defined
\beq
	\{F,G\}_{\mathfrak{g}^*}(\alpha) \equiv \alpha\bigg([\lambda(\delta F),\lambda(\delta G)]_{\mathfrak{g}}\bigg), \; F,G \in C^{\infty}(\mathfrak{g}^*). 
\eeq 

When $\mathfrak{g}$ is finite dimensional, it is straightforward to show that $X \times_{a} G$ is locally isomorphic to $X \times G$.\footnote{This is also true in the infinite dimensional case provided $X \times_a G$ is taken to be locally modeled on $X \times \mathfrak{g}^{\musSharp}$, where $\mathfrak{g}^{\musSharp}$ is an appropriate subset of $\mathfrak{g}^*$ that is induced from a bilinear on $\mathfrak{g}$. For example, this may come from the Killing form if $\mathfrak{g}$ is semi-simple, although we don't require the bilinear to be invariant for this purpose.} In fact, $X \times_a G$ has the structure of a principal bundle with base $X$ and structure group $G$. The projection is as defined above, while the right action is inherited from the symmetry action of $G$ on $X$ and the standard right action of $G$ on itself. In particular,
\beq
	R_h(x,g) \equiv (a_{h^{-1}}(x), gh). 
\eeq	

In \cite{Klinger:2023tgi}, it was argued that the extended phase space is the classical analog of the crossed product algebra. For the benefit of the reader, we recall here the conventional definition of the crossed product which will play a central role in the remainder of the note. A von Neumann covariant system is a triple $(M,G,\alpha)$ with $M$ a von Neumann algebra and $\alpha: G \rightarrow \text{Aut}(M)$ an automorphism action on $M$ by a locally compact group $G$. To each such triple we can identify an enlarged von Neumann algebra $M \times_{\alpha} G$ called the crossed product \cite{takesaki1973duality}. 

Let $\pi: M \rightarrow B(H)$ be a representation of $M$. The algebraic structure of the crossed product is defined by means of the canonical covariant representation $(H_G \equiv L^2(G,H), \pi_{\alpha}, \lambda)$. Here $H_G$ is called the extended Hilbert space, and $\pi_{\alpha}: M \rightarrow B(H_G)$ and $\lambda: G \rightarrow U(H_G)$ are respectively a representation and unitary representation given by
\beq
	\bigg(\pi_{\alpha}(x) \xi\bigg)(g) \equiv \pi \circ \alpha_{g^{-1}}(x) \bigg(\xi(g)\bigg), \qquad \bigg(\lambda(h) \xi\bigg)(g) \equiv \xi(h^{-1} g). 
\eeq	
The triple $(H_G, \pi_{\alpha},\lambda)$ is called covariant because $\lambda$ unitarily implements the automorphism $\alpha$ as
\beq
    \pi_{\alpha} \circ \alpha_g(x) = \lambda(g) \pi_{\alpha}(x) \lambda(g^{-1}).
\eeq
The covariant representation identifies a closed algebra in $B(H_G)$ consisting of elements of the form $\pi_{\alpha}(x) \lambda(g)$ with $g \in G$ and $x \in M$. To be precise, we have
\beq
	\pi_{\alpha}(x) \pi_{\alpha}(y) = \pi_{\alpha}(xy), \qquad \lambda(g) \lambda(h) = \lambda(gh), \qquad \lambda(g) \pi_{\alpha}(x) = \pi_{\alpha} \circ \alpha_g(x) \lambda(g),
\eeq
or more succinctly
\beq \label{Quantum Poisson Bracket}
	\bigg(\pi_{\alpha}(x_1) \lambda(g_1) \bigg) \bigg(\pi_{\alpha}(x_2) \lambda(g_2)\bigg) = \pi_{\alpha}\bigg(x_1 \alpha_{g_1}(x_2)\bigg) \lambda(g_1 g_2). 
\eeq
The crossed product is defined to be the von Neumann algebra generated by this representation: $M \times_{\alpha} G \equiv \pi_{\alpha}(M) \vee \lambda(G)$. 

Eqn. \eqref{Quantum Poisson Bracket} is a quantum analog of \eqref{extended bracket}, with the operator product structure on the crossed product specified in terms of (i) the product structure on $M$ (which quantizes the Poisson bracket on $X$), (ii) the product structure on $G$ (which quantizes the canonical Poisson bracket on $\mathfrak{g}^*)$, and (iii) the action of $G$ on $M$ (which quantizes the infinitesimal action of $G$ on $X$). One of the immediate objectives of this note is to place this correspondence on firm mathematical footing.

\subsection{Quantizing the Principal Bundle} \label{sec: Intro QPB}

We would now like to `quantize' the notion of a principal bundle and in turn that of the extended phase space. In this paper, the notion of quantization we will use is one that comes to us from non-commutative geometry. Quantization in non-commutative geometry roughly follows in a sequence of two stages. First, classical geometric definitions are reformulated in algebraic terms by translating statements about spaces into statements about algebras of functions therein. Second, the resulting statements are promoted to definitions and allowed to hold even when the algebras involved are not commutative function algebras. 

To get a sense of how this procedure works, let us reformulate the definition of a classical principal bundle in this way. To do so, we simply pullback the structure and trivialization maps. In this language, a principal bundle is a triple of commutative algebras $(C^{\infty}(P), C^{\infty}(B), C^{\infty}(G))$ along with a pair of maps $\pi^*: C^{\infty}(B) \hookrightarrow C^{\infty}(P)$ and $R^*: C^{\infty}(P) \rightarrow C^{\infty}(P) \times C^{\infty}(G)$. The principal bundle is called trivial\footnote{We postpone a discussion of the meaning of `local' trivialization in the non-commutative context to Section \ref{sec: locally trivial}.} if it admits a trivialization map $t^*: C^{\infty}(G) \hookrightarrow C^{\infty}(P)$.

That $\pi$ was a projection implies that $\pi^*$ is an inclusion. To understand the algebraic translation of the freeness condition, it is useful to reformulate it in the following way. First, let us define the `multiplication' $m: P \rightarrow P \times P$, with $m(p) \equiv (p,p)$. Then, we can define the map $\Phi \equiv (id_P \times R) \circ (m \times id_G): P \times G \rightarrow P \times P$ such that
\beq \label{Freeness Map}
	\Phi(p,g) = (id_P \times R) \circ (m \times id_G)(p,g) = (id_P \times R)(p,p,g) = (p,R_g(p)). 
\eeq
From this point of view, freeness simply implies the map $\Phi$ is injective for each fixed $p$. Recalling that the pullback of an injective map is surjective, the algebraic reformulation of the freeness condition is that
\beq
	\Phi^* = (m^* \times id_G) \circ (id_P \times R^*): C^{\infty}(P) \times C^{\infty}(P) \rightarrow C^{\infty}(P) \times C^{\infty}(G)
\eeq
is surjective. To upgrade the trivialization condition it is useful to rewrite \eqref{trivialization condition} in the form
\beq \label{trivialization condition 2}
	t \circ R = r \circ t \times id_G,
\eeq
where here $r: G \times G \rightarrow G$ is the right product in the group, $r(g,h) = hg$. Pulling back \eqref{trivialization condition 2} we find
\beq \label{trivialization condition 3}
	R^* \circ t^* = (t^* \times id_G) \circ r^*.
\eeq

We can now upgrade this definition to a fully non-commutative form by elevating the spaces involved to the status of arbitrary algebras. In so doing we will also have to adjust the definitions of some of the maps accordingly. The function spaces $C^{\infty}(P)$ and $C^{\infty}(B)$ are promoted to unital algebras which we denote by $\mathcal{P}$ and $\mathcal{B}$. The pullback of the multiplication map $m(p) = (p,p)$ is interpreted as the algebra product and denoted by $m_{\mathcal{P}}$. Since $C^{\infty}(G)$ is the function algebra of a group, it must have some additional structure. It turns out that the appropriate generalization is that of a quantum group, also known as a Hopf algebra. In Appendix \ref{app: Hopf} we introduce the necessary background on Hopf algebras. We denote such an algebra by $\mathcal{G}$. The pullback of the right group multiplication is replaced by the comultiplication $\delta_{\mathcal{G}}$. Finally, the appropriate generalization of a right action to the non-commutative context is the notion of a right coaction. 

Putting everything together, we arrive at the following definition of a trivial quantum principal bundle:

\begin{definition} \label{quantum principal bundle}
	A quantum principal bundle is a triple $(\mathcal{P},\mathcal{B},\mathcal{G})$ such that $\mathcal{P}$ and $\mathcal{B}$ are unital algebras and $\mathcal{G}$ is a quantum group. $\mathcal{P}$ is a right comodule algebra of $\mathcal{G}$ admitting a right coaction $R^{\mathcal{P}}: \mathcal{P} \rightarrow \mathcal{P} \otimes \mathcal{G}$. The coaction is free in the sense that
	\beq
		(m_{\mathcal{P}} \otimes id_{\mathcal{G}}) \circ (id_{\mathcal{P}} \otimes R^{\mathcal{P}}): \mathcal{P} \otimes \mathcal{P} \rightarrow \mathcal{P} \otimes \mathcal{G}
	\eeq
	is surjective. This quantizes the condition that \eqref{Freeness Map} is injective. The algebra $\mathcal{B}$ can be identified with the fixed point subalgebra of $\mathcal{P}$ under the right coaction,
	\beq
		\mathcal{B} = \{p \in \mathcal{P} \; | \; R^{\mathcal{P}}(p) = p \otimes \mathbb{1}_{\mathcal{G}}\}.
	\eeq 
	We therefore automatically obtain an inclusion $\Pi: \mathcal{B} \hookrightarrow \mathcal{P}$. This quantizes the projection $\pi: P \rightarrow M$. Finally, the quantum principal bundle $\mathcal{P}$ is termed trivial if it admits a unital, convolution invertible map $T: \mathcal{G} \hookrightarrow \mathcal{P}$ such that
	\beq
		R^{\mathcal{P}} \circ T = (T \otimes id_{\mathcal{G}}) \circ \delta_{\mathcal{G}}.
	\eeq
	This quantizes \eqref{trivialization condition 3}.
\end{definition}

\subsection{Crossed Products are Trivial QPBs} \label{sec: CP = QPB}

We are now prepared to compare the ensuing discussion with Landstad's construction of the crossed product \cite{landstad1979duality}. According to Landstad's theorem, given a von Neumann algebra $\mathcal{P}$ and a locally compact group $G$ the algebra $\mathcal{P}$ is isomorphic to a crossed product $\mathcal{B} \times_{\alpha} G$ for some covariant system $(\mathcal{B},G,\alpha)$ if and only if there exists a continuous homomorphism $\lambda: G \rightarrow \mathcal{P}$ and a coaction $R^{\mathcal{P}} \rightarrow \mathcal{P} \otimes \mathcal{L}(G)$ such that
\beq \label{Landstad 1}
	(R^{\mathcal{P}} \otimes id_{\mathcal{L}(G)}) \circ R^{\mathcal{P}} = (id_{\mathcal{P}} \otimes \delta_G) \circ R^{\mathcal{P}}, \qquad R^{\mathcal{P}} \circ \lambda(g) = \lambda(g) \otimes \ell(g), \; \forall g \in G. 
\eeq
Here $\mathcal{L}(G)$ is the group von Neumann algebra of $G$ which is given the structure of a coalgebra by endowing it with the comultiplication $\delta_G: \mathcal{L}(G) \rightarrow \mathcal{L}(G)^{\otimes 2}$ such that $\delta_G(\ell(g)) = \ell(g) \otimes \ell(g)$. In fact, $\mathcal{L}(G)$ moreover has the structure of a Hopf algebra. Recall that, in the event that $\mathcal{P}$ is isomorphic to a crossed product, the system algebra is defined to be the fixed point subalgebra of $\mathcal{P}$ with respect to the coaction $R^{\mathcal{P}}$:
\beq
	\mathcal{B} \equiv \{p \in \mathcal{P} \; | \; R_p(p) = p \otimes \mathbb{1}_{\mathcal{L}(G)}\}. 
\eeq

The first of the two equations in \eqref{Landstad 1} identifies $\mathcal{P}$ as a right $\mathcal{L}(G)$ comodule algebra. To parse the second equation in \eqref{Landstad 1} let's reinterpret the map $\lambda: G \rightarrow \mathcal{P}$ as a map $T: \mathcal{L}(G) \hookrightarrow \mathcal{P}$ such that $T(\ell(g)) \equiv \lambda(g)$. Then we can rewrite 
\beq
	R^{\mathcal{P}} \circ T(\ell(g)) = \lambda(g) \otimes \ell(g) = (T \otimes id_{\mathcal{L}(G)}) \circ \delta_G(\ell(g)), \; \forall g \in G. 
\eeq	
As $\ell(g)$ is dense in $\mathcal{L}(G)$ this implies that
\beq
	R^{\mathcal{P}} \circ T = (T \otimes id_{\mathcal{L}(G)}) \circ \delta_G.
\eeq
Provided $T$ is convolution invertible and $T(\ell(e)) = \mathbb{1}_{\mathcal{P}}$ the second condition in \eqref{Landstad} therefore identifies $(\mathcal{P},\mathcal{B},\mathcal{L}(G))$ as a trivial quantum principal bundle. In fact, this is always the case.

We formalize the above observations in the form of the following theorem: 

\begin{theorem}[vN QPB with group structure $\iff$ vN Crossed Product]
The following are equivalent:
\begin{enumerate}
	\item $(\mathcal{P},\mathcal{B},\mathcal{G})$ is a trivial quantum principal bundle with $\mathcal{P}$ a von Neumann algebra and $\mathcal{G} = \mathcal{L}(G)$ for $G$ a locally compact group.
	\item $\mathcal{P}$ is a von Neumann algebra for which there exists a von Neumann covariant system $(\mathcal{B},G,\alpha)$ such that $\mathcal{P} \simeq \mathcal{B} \times_{\alpha} G$. The automorphism action of $G$ on $\mathcal{B}$ is given by $\alpha_g(b) = T(\ell(g)) b T(\ell(g^{-1}))$.  
\end{enumerate}
\end{theorem}

\begin{proof}
	The implication $1 \implies 2$ has been discussed above. $\mathcal{P}$ is a right comodule algebra of $\mathcal{L}(G)$ which implies that it admits a right coaction, $R^{\mathcal{P}}$, of $\mathcal{L}(G)$ satisfying the conditions of Landstad's theorem. At the same time the trivialization map $T: \mathcal{L}(G) \hookrightarrow \mathcal{P}$ can be used to define a homomorphism $\lambda: G \rightarrow \mathcal{P}$ by $\lambda(g) \equiv T(\ell(g))$. The right equivariance of $T$ under the coaction $R^{\mathcal{P}}$ implies that the homomorphism $\lambda(g)$ satisfies the condition $R^{\mathcal{P}}(\lambda(g)) = \lambda(g) \otimes \ell(g)$ for each $g \in G$. Thus, both of the conditions of Landstad's theorem are satisfied and we conclude that there exists a covariant system $(\mathcal{B},G,\alpha)$ such that $\mathcal{P} \simeq \mathcal{B} \times_{\alpha} G$. Here, $\mathcal{B} = \{b \in \mathcal{P} \: | \; R^{\mathcal{P}}(p) = p \otimes \mathbb{1}_{\mathcal{L}(G)}\}$ and $\alpha_g = \text{Ad}_{\lambda(g)} = \text{Ad}_{T(\ell(g))}$. 
	
	Let us now consider the opposite direction, $2 \implies 1$. Suppose that $\mathcal{P}$ is a von Neumann algebra which is isomorphic to the crossed product $\mathcal{B} \times_{\alpha} G$ induced by some covariant system $(\mathcal{B},G,\alpha)$. By Landstad's theorem this implies the existence of maps $R^{\mathcal{P}}: \mathcal{P} \rightarrow \mathcal{P} \otimes \mathcal{L}(G)$ and $\lambda: G \rightarrow \mathcal{P}$, both of which are homomorphisms, satisfying
	\beq \label{Landstad}
		R^{\mathcal{P}}(\lambda(g)) = \lambda(g) \otimes \ell(g) \; \forall g \in G, \qquad (R^{\mathcal{P}} \otimes id_{\mathcal{L}(G)}) \circ R^{\mathcal{P}} = (id_{\mathcal{P}} \otimes \delta_G) \circ R^{\mathcal{P}}. 
	\eeq	
	In eqn. \eqref{Landstad}, $\delta_G(\ell(g)) \equiv \ell(g) \otimes \ell(g)$ defines a comultiplication on $\mathcal{L}(G)$. The algebra $\mathcal{L}(G)$ can moreover be completed to a Hopf algebra by endowing it with the counit $\epsilon_G: \mathcal{L}(G) \rightarrow \mathbb{C}$ by virtue of the Plancherel weight. That is
	\beq
		\epsilon_G(f) = f(e), \; \epsilon_{G}(\ell(g)) = 1, \; f \in C_0(G), g \in G. 
	\eeq
	Finally, Landstad identifies the algebra $\mathcal{B} \equiv \{b \in \mathcal{P} \; | \; R^{\mathcal{P}}(b) = b \otimes \mathbb{1} \}$ and the automorphism $\alpha_g = \text{Ad}_{\lambda(g)}$. 
	
	By the isomorphism $\mathcal{P} \simeq \mathcal{B} \times_{\alpha} G$, the set of products $b \lambda(g)$ is dense in $\mathcal{P}$. Eqn. \eqref{Landstad} implies that
	\beq
		(id_{\mathcal{P}} \otimes \epsilon_G) \circ R^{\mathcal{P}}(b \lambda(g)) = id_{\mathcal{P}}(b \lambda(g)) \otimes \epsilon_G(\ell(g)) = b \lambda(g),
	\eeq
	which demonstrates that $R^{\mathcal{P}}$ is a right coaction and therefore $\mathcal{P}$ is a right $\mathcal{L}(G)$ comodule algebra. We can also define the map $T: \mathcal{L}(G) \hookrightarrow \mathcal{P}$ by $T(\ell(g)) \equiv \lambda(g)$ (and extended to the rest of $\mathcal{L}(G)$ by closure). The map $T$ is clearly unital
	\beq
		T(\mathbb{1}_{\mathcal{L}(G)}) = T(\ell(e)) = \lambda(e) = \mathbb{1}_P.  
	\eeq
	It is also convolution invertible; its convolutional inverse is given by $T^{-1}: \mathcal{L}(G) \rightarrow \mathcal{P}$ with $T^{-1}(\ell(g)) \equiv \lambda(g)^{-1}$. Then
	\beq \label{Convolutional Invertibility}
		T \star T^{-1}(\ell(g)) = m_{\mathcal{P}} \circ (T \otimes T^{-1}) \circ \delta_G(\ell(g)) = T(\ell(g))T^{-1}(\ell(g)) = \lambda(g) \lambda(g)^{-1} = \mathbb{1}_{\mathcal{P}}, \; \forall g \in G.  
	\eeq	
	Since $\mathcal{P}$ is a von Neumann algebra $\eta_{\mathcal{P}}(z) = z \mathbb{1}_P$ and thus the convolutional identity $\eta_{\mathcal{P}} \circ \epsilon_G(\ell(g)) = \eta_{\mathcal{P}}(1) = \mathbb{1}_{\mathcal{P}}$ for each $g$. Therefore eqn. \eqref{Convolutional Invertibility} can be rewritten as
	\beq
		T \star T^{-1} = \eta_{\mathcal{P}} \circ \epsilon_G,
	\eeq
	which confirms that $T^{-1}$ is the convolutional inverse of $T$. We have therefore shown that any von Neumann algebra $\mathcal{P}$ that is isomorphic to a crossed product admits a coaction $R^{\mathcal{P}}$ and a trivialization map $T$ and may therefore be regarded as a trivial quantum principal bundle. The base algebra of this quantum principal bundle is the system algebra $\mathcal{B}$. 
\end{proof}

\section{Quantum Orbifolds and Quantum Reference Frames} \label{sec: Relational QPBs}

We have now demonstrated that von Neumann crossed product algebras are quantum principal bundles. In \cite{ahmad2024quantumreferenceframestopdown}, we illustrated how crossed product algebras could be identified with quantum reference frames. In this section, we will explore how these two points of view meld together to create a non-commutative geometry for quantum reference frames. In particular, we will explore the non-commutative analog of gauge transformations and show that they are reference frame preserving maps. This provides a powerful new tool for mapping and distinguishing between quantum reference frames. To this end, we work out the consequences of a quantum gauge transformation on the constraints encountered by a particular observer. Finally, we describe the non-commutative notion of local trivialization, and comment on its relationship with the quantum relativity of reference frames. 

\subsection{Quantum Gauge Transformations} \label{sec: QGT}

In the case of the quantum gauge transformation it is more transparent to start with the general definition and demonstrate its classical limit than the other way around.

\begin{definition}[Quantum gauge transformation] \label{gauge transformation}
	Let $(\mathcal{P},\mathcal{B},\mathcal{G})$ be a trivial quantum principal bundle with trivialization $T: \mathcal{G} \hookrightarrow \mathcal{P}$. A quantum gauge transformation is a unital, convolution invertible map $\Gamma: \mathcal{G} \rightarrow \mathcal{B}$ which acts upon the trivialization in the following way:
	\beq
		T \mapsto T^{\Gamma} \equiv \bigg(\Pi \circ \Gamma\bigg) \star T. 
	\eeq
	Here, $\star$ is the convolutional product in $\mathcal{L}(\mathcal{G},\mathcal{B})$. 
\end{definition}

In the classical case, we have $T = t^*$ for a trivialization map $t: P \rightarrow G$. The quantum gauge transformation can be interpreted as a section $\gamma: B \rightarrow G$ giving rise to the map $\Gamma = \gamma^*: C^{\infty}(G) \rightarrow C^{\infty}(B)$. Identifying $C^{\infty}(G) \otimes C^{\infty}(G) \simeq C^{\infty}(G \times G)$, the multiplication and comultiplication on $C^{\infty}(G)$ are defined respectively by
\beq
	m_{C^{\infty}(G)}(F_1)(g) \equiv F_1(g,g), \qquad \delta_{C^{\infty}(G)}(F_2)(g,h) \equiv F_2(gh), \; F_1 \in C^{\infty}(G \times G), F_2 \in C^{\infty}(G).
\eeq
Given $F \in C^{\infty}(G)$ we therefore obtain
\begin{flalign}
	T^{\Gamma}(F)\rvert_{p} &= \bigg(\pi^* \circ \gamma^*\bigg) \star t^*(F)\rvert_p \nonumber \\
	&= m_{C^{\infty}(G)} \circ \bigg((\gamma \circ \pi)^* \otimes t^*\bigg) \circ \delta_{C^{\infty}(G)}(F)\rvert_p \nonumber \\
	&= m_{C^{\infty}(G)} \bigg(\delta_{C^{\infty}(G)}(F)(\gamma \circ \pi(p), t(p))\bigg) \nonumber \\
	&= F\bigg(\gamma \circ \pi(p) t(p)\bigg) = \bigg(t_{\gamma}^*F\bigg)(p). 
\end{flalign}
In the last line we have identified $T^{\Gamma} = t_{\gamma}^*$ where $t_{\gamma}: P \rightarrow G$ is the gauge transformed trivialization $t_{\gamma}(p) \equiv \gamma \circ \pi(p) t(p)$. Of course, this is precisely the standard formula for the gauge transformation as implemented by the left action of the group on the trivialization. 

Our main interest in this section is introducing a relational interpretation for the quantum gauge transformation. In \cite{ahmad2024quantumreferenceframestopdown}, we defined a quantum reference frame as a von Neumann covariant system $(\mathcal{B},G,\alpha)$ giving rise to a crossed product algebra $\mathcal{P} \equiv \mathcal{B} \times_{\alpha} G$. As we have now established, $(\mathcal{P},\mathcal{B},\mathcal{L}(G))$ defines a trivial quantum principal bundle. Moreover, the group action is related to the trivialization map by the equation $\alpha_g \equiv \text{Ad}_{T(\ell(g))}$. Thus, in this context, a quantum gauge transformation changes the automorphism action:
\beq
	\alpha \mapsto \alpha^{\Gamma}, \qquad \alpha^{\Gamma}_g \equiv \text{Ad}_{T^{\Gamma}(\ell(g))}. 
\eeq
A priori, this change gives rise to a distinct covariant system $(\mathcal{B},G,\alpha^{\Gamma})$ and by extension a new quantum reference frame. However, as we will now demonstrate, a quantum gauge transformation can be interpreted as an inner conjugation of $\alpha$. Thus, it leaves both the covariant system and its associated crossed product algebra invariant. In other words, a quantum gauge transformation is precisely a QRF preserving map between crossed product algebras. 

Let $\pi: \mathcal{B} \rightarrow B(H)$ be a representation of $\mathcal{B}$, and denote by $(H_G, \lambda, \pi_{\alpha})$ the canonical covariant representation of the covariant system $(\mathcal{B},G,\alpha)$. The crossed product is given by the von Neumann union of these representations: $\mathcal{P} \equiv \pi_{\alpha}(\mathcal{B}) \vee \lambda(G)$. The map $\pi_{\alpha}: \mathcal{B} \hookrightarrow \mathcal{P}$ is an inclusion which should be interpreted as the `projection' associated with $\mathcal{P}$ when regarded as a trivial quantum principal bundle. Similarly, the map $\lambda$ gives rise to the trivialization $T: \mathcal{L}(G) \rightarrow \mathcal{P}$ given by $T(\ell(g)) \equiv \lambda(g)$. In this context, a quantum gauge transformation is a unital, convolution invertible map $\Gamma: \mathcal{L}(G) \rightarrow \mathcal{B}$. On $\mathcal{L}(G)$, the multiplication $m_{\mathcal{L}(G)}: \mathcal{L}(G) \otimes \mathcal{L}(G)$ and comultiplication $\delta_{\mathcal{L}(G)}: \mathcal{L}(G) \rightarrow \mathcal{L}(G) \otimes \mathcal{L}(G)$ are given respectively by
\beq
	m_{\mathcal{L}(G)}(\ell(g) \otimes \ell(h)) \equiv \ell(gh), \qquad \delta_{\mathcal{L}(G)}(\ell(g)) \equiv \ell(g) \otimes \ell(g), \; g,h \in G. 
\eeq	
Thus, for any $f_1,f_2 \in \mathcal{L}(\mathcal{L}(G),\mathcal{B})$ the convolutional product is simply the pointwise multiplication
\beq
	\bigg(f_1 \star f_2\bigg)(\ell(g)) \equiv m_{\mathcal{L}(G)} \circ (f_1 \otimes f_2) \circ \delta_{\mathcal{L}(G)}(\ell(g)) = f_1(\ell(g)) f_2(\ell(g)). 
\eeq
In light of this fact, a gauge transformation is equivalent to a map $v: G \rightarrow \mathcal{B}$ such that $\Gamma(\ell(g)) \equiv v(g)$, each $v(g)$ is invertible, and $v(e) = \mathbb{1}$.

Applying the gauge transformation directly to the trivialization $T$ we find
\beq
	T^{\Gamma}(\ell(g)) = \bigg(\pi_{\alpha} \circ \Gamma \bigg) \star T(\ell(g)) = \pi_{\alpha} \circ \Gamma(\ell(g)) T(\ell(g)) = \pi_{\alpha}(v(g)) \lambda(g). 
\eeq
In other words, the effect of the gauge transformation is to deform the covariant representation
\beq
	\lambda \mapsto \lambda^{\Gamma}, \; \lambda^{\Gamma}(g) = \pi_{\alpha}(v(g)) \lambda(g). 
\eeq
Notice that $\lambda^{\Gamma}(g)$ is constructed from the product of elements in $\pi_{\alpha}(\mathcal{B})$ and $\lambda(G)$, and therefore lives inside of $\mathcal{P}$. Indeed, the algebra generated by $\pi_{\alpha}$ and the gauge transformed group representation is isomorphic to the original crossed product: $\pi_{\alpha}(\mathcal{B}) \vee \lambda^{\Gamma}(G) \simeq \mathcal{P}$. More to the point, the gauge transformed automorphism is given by
\beq
	\alpha^{\Gamma}_g = \text{Ad}_{\lambda^{\Gamma}(g)} = \text{Ad}_{\pi_{\alpha}(v(g))} \circ \alpha_g,
\eeq
which is merely a conjugation of the original automorphism by elements inside of $\mathcal{B}$. This implies that $\alpha^{\Gamma}$ and $\alpha$ belong to the same equivalence class of outer automorphisms. An immediate corollary of this fact is that $(\mathcal{B},G,\alpha)$ and $(\mathcal{B},G,\alpha^{\Gamma})$ are equivalent as covariant systems, and thus as quantum reference frames. This suggests an intimate relation between cohomology and QRF transformations. 

\subsection{Constraint Quantization} \label{sec: constraint quant}

To clarify its relational interpretation, let us consider the quantum gauge transformation in the context of a simple constraint quantization problem. Recall that by the commutation theorem~\cite{van1978continuous}, any crossed product $\mathcal{P} \equiv \mathcal{B} \times_{\alpha} G$ may equivalently be viewed as the fixed-point subalgebra of $\mathcal{B} \otimes B(L^2(G))$ under an extended automorphism $\hat{\alpha} = \alpha \otimes \text{Ad}_{r}$:
\begin{equation} \label{Commutation theorem}
    \mathcal{P} = \mathcal{B} \times_{\alpha} G = \left( \mathcal{B} \otimes B(L^{2}(G))\right)^{\alpha \otimes \text{Ad}_{r}}.
\end{equation}
Here $r: G \rightarrow U(L^2(G))$ is the right regular representation. 

For illustration purposes, assume the group is one-dimensional and so has a single generator $\hat{C}_{G}$. Let $\pi: \mathcal{B} \rightarrow B(H)$ be a representation. Then, there exists an operator $\hat{C}_{\mathcal{B}} \in B(H)$ generating the action of $\alpha$ such that \eqref{Commutation theorem} may be interpreted as a constraint quantization of the system $\mathcal{B}$ by appending a quantum reference frame whose Hilbert space is $L^{2}(G)$ along with the total constraint
\begin{equation} \label{constraint}
    \hat{C} \equiv \hat{C}_{\mathcal{B}} + \hat{C}_{G}. 
\end{equation}
The commutation theorem tells us that the crossed product algebra is the commutant of $\hat{C}$ on the Hilbert space $H \otimes L^{2}(G)$.\footnote{We should emphasize that our approach to constraint implementation ensues at the operator algebraic level, as opposed to the Hilbert space level. This is relevant, for example, to projective representations as appearing in \cite{Ciambelli:2024swv}. }

Let $(H_G \equiv H \otimes L^2(G), \pi_{\alpha}, \lambda)$ be the canonical covariant representation of the crossed product $\mathcal{P}$. Suppose that we choose an operator $b \in \mathcal{B}$ such that $T_{b} \equiv \pi_{\alpha}(b) \in \mathcal{P}$ is invertible. Then, we can conjugate the above constraint with this operator to obtain
\begin{equation}
    \hat{C}_{b} \equiv T_{b} \hat{C} T_{b}^{-1}.
\end{equation}
Now, let us consider the impact of conjugating the constraint on the definition of the crossed product. Suppose that $A \in B(H_G)$ is an operator which commutes with the conjugated constraint:
\beq
	0 = [A,\hat{C}_b] = A \hat{C}_b - \hat{C}_b A. 
\eeq
Expanding this relation and using the invertibility of $T_b$ we find:
\beq
	[T_b^{-1} A T_b, \hat{C}] = 0. 
\eeq
That is, if $A$ belongs to the crossed product defined by the conjugated constraints, then $A_b \equiv T_b^{-1} A T_b$ belongs to the original crossed product. Since $T_b$ is assumed invertible, this transformation is one to one and thus we conclude that the gauge-invariant algebras defined by either $\hat{C}$ or $\hat{C}_b$ are in fact isomorphic. 

This is the physical manifestation of the observation $\hat{C}_{b}$ is an inner conjugate of $\hat{C}$. Indeed, this is just the realization that the crossed product $\mathcal{P}$ is only sensitive to the outer part of the action $\alpha$. The conjugation by an invertible element does not affect the gauge-invariant physics, but it does affect the \textit{presentation} of the constraint quantization problem. Naturally, we interpret such transformations of a dynamical system that change the constraints but not the physics as gauge transformations. Since, the frame representation itself is augmented by an element of the original system $\mathcal{B}$ we can also view them as frame transformations. In summary, innerly conjugating the system and frame representations appropriately does not change the physics but makes the dynamical system appear to have a new frame attached to it.

\subsection{Locally Trivial Quantum Principal Bundles and G-Framed Algebras} \label{sec: locally trivial}

Up to this point we have concentrated our discussion on trivial quantum principal bundles. An important lingering question is therefore what it means to have a locally trivial principal bundle in the non-commutative context. As we shall now argue, local trivialization corresponds to the existence of multiple distinct quantum reference frames. From this point of view, quantum gauge transformations play the natural role of quantifying change of frame data relating isomorphic or partially isomorphic QRFs. Overall, the structure of a locally trivial quantum principal bundle very closely mirrors that of a $G$-framed algebra, as introduced in \cite{ahmad2024quantumreferenceframestopdown}. In fact, the G-framed algebra is a more general structure which we term a quantum orbifold. In this way, the insights of our previous work are given a formal mathematical language in the arena of non-commutative geometry. 

We start by reviewing what is referred to as a locally trivial quantum principal bundle in the literature~\cite{Brzezinski:1992hda, budzynski1994quantumprincipalfiberbundles, Calow_2002}. 
\begin{definition}[Locally Trivial Quantum Principal Bundle]
    A locally trivial quantum principal bundle is a triple $(\mathcal{P}, \mathcal{B}, \mathcal{G})$ consisting of unital algebras $\mathcal{P}$ and $\mathcal{B}$, and quantum group $\mathcal{G}$, along with a collection of algebras $\mathcal{A}_{\mathcal{P}} \equiv \{\mathcal{P}_i\}_{i \in \mathcal{I}}$ indexed by a partially ordered set $\mathcal{I}$. The set $\mathcal{A}_{\mathcal{P}}$ can be thought of as an atlas for the algebra $\mathcal{P}$ in the sense that, for each $i \geq j$, there exists a restriction map $r_{ij}: \mathcal{P}_i \rightarrow \mathcal{P}_j$. In particular, there is always such a restriction map from $\mathcal{P}$ to any $\mathcal{P}_i$, which we denote by $r_i: \mathcal{P} \rightarrow \mathcal{P}_j$. Each $\mathcal{P}_i$ is moreover a $\mathcal{G}$-comodule algebra with free right action $R^{\mathcal{P}_i}: \mathcal{P}_i \rightarrow \mathcal{P}_i \otimes \mathcal{G}$, so that we can define $\mathcal{B}_i \equiv \{b_i \in \mathcal{P}_i \; | R^{\mathcal{P}_i}(b_i) = b_i \otimes \mathbb{1} \}$, and the collection $\mathcal{A}_{\mathcal{B}} \equiv \{\mathcal{B}_i\}_{\mathcal{I}}$ is an atlas for $\mathcal{B}$. We denote by $\Pi_i: \mathcal{B}_i \hookrightarrow \mathcal{P}_i$ the inclusion $\mathcal{B}_i \subset \mathcal{P}_i$. In each `chart' $i$ there exists a trivialization map $T_i: \mathcal{G} \rightarrow \mathcal{P}_i$ so that $(\mathcal{P}_i, \mathcal{B}_i, \mathcal{G})$ has the structure of a trivial quantum principal bundle. Finally, for every $i \geq j$ there exists a quantum gauge transformation $\Gamma_{ij}: \mathcal{G} \rightarrow \mathcal{B}_{ij}$ such that $T_i = \bigg(\Pi_{ij} \circ \Gamma_{ij}\bigg) \star T_j$, and $\Gamma_{ij}$ satisfy a triple overlap condition $\Gamma_{ij} \star \Gamma_{jk} = \Gamma_{ik}$ for all $i \geq j \geq k$. 
\end{definition}

While the above definition seems abstruse, the reader is instructed to keep in mind that the basic idea is a stitching together of several trivial quantum principal bundles per Definition~\eqref{quantum principal bundle} via quantum gauge transformations per Definition~\eqref{gauge transformation}. Given that each trivial quantum principal bundle in our context is a crossed product associated to a dynamical system $(\mathcal{B}_{i}, G, \alpha_{i})$, we see that the above structure realizes a superstructure $\mathcal{P}$ covered by several dynamical systems which we interpreted in our previous work as quantum reference frames~\cite{ahmad2024quantumreferenceframestopdown}.  

The issue with the above definition in the relational context is the partial order $i \leq j$. This seems perfectly natural from the perspective of quantizing classical geometry, where the elements of the cover are included in each other via simple subset inclusion. However, in our case:
\begin{enumerate}
    \item There is in general no way to compare the degrees of freedom housed in inequivalent quantum reference frames and so one does not expect a partial order.
    \item Relatedly, one cannot consistently take products between operators as seen from inequivalent quantum reference frames and so one does not expect a product.
\end{enumerate}
The above two complaints are very physical in nature: there could be two frames which completely do not overlap and so access disjoint degrees of freedom of the system of interest. As a corollary, the system algebras $\mathcal{B}_{i}$ associated to each of these frames will not contain any isomorphic subalgebras and given the Gribov ambiguity, may even have different symmetries. Thus, to accurately model the superstructure which subsumes multiple inequivalent QRFs (with potential overlaps glued by QRF transformations) one is really in search of a quantum orbifold rather than a bundle. These considerations motivate the following definition.
\begin{definition}[Quantum Orbifold]
	A quantum orbifold is a triple $(\mathcal{P}, \mathcal{B}, \mathcal{G})$ consisting of unital algebras $\mathcal{P}$ and $\mathcal{B}$, and quantum group $\mathcal{G}$, along with a set $\mathcal{I}$ indexing collections of algebras $\mathcal{A}_{\mathcal{P}} \equiv \{\mathcal{P}_i\}_{i \in \mathcal{I}}$, and  quantum groups $\mathcal{A}_{\mathcal{G}} \equiv \{\mathcal{G}_i\}_{i \in \mathcal{I}}$, together with a subset $\mathcal{I}^{(2)} \subset \mathcal{I} \times \mathcal{I}$ called the intersection. The collections $\mathcal{A}_{\mathcal{P}}$ and $\mathcal{A}_{\mathcal{G}}$ play the role of atlases for $\mathcal{P}$ and $\mathcal{G}$, respectively, with the caveat that there exist restriction maps $r^{\mathcal{P}}_{ij}: \mathcal{P}_i \rightarrow \mathcal{P}_j$ if and only if $(i,j) \in \mathcal{I}^{(2)}$ (and likewise for $\mathcal{G}$). In this case we denote by $\mathcal{P}_{ij}$ the `intersection' of $\mathcal{P}_i$ and $\mathcal{P}_j$; e.g. for each $p_{ij} \in \mathcal{P}_{ij}$ there exist $p_i \in \mathcal{P}_i$ and $p_j \in \mathcal{P}_j$ such that $r^{\mathcal{P}}_{i,ij}(p_i) = r^{\mathcal{P}}_{j,ij}(p_j) = p_{ij}$. For each $i$ we require that $\mathcal{P}_i$ is a right $\mathcal{G}_i$-comodule algebra with free right action $R^{\mathcal{P}_i}: \mathcal{P}_i \rightarrow \mathcal{P}_i \otimes \mathcal{G}_i$. We denote by $\mathcal{B}_i \equiv \{b_i \in \mathcal{P}_i \; | \; R^{\mathcal{P}_i}(b_i) = b_i \otimes \mathbb{1}\}$, and $\Pi_i: \mathcal{B}_i \hookrightarrow \mathcal{P}_i$ the associated inclusion, so that $\mathcal{A}_{\mathcal{B}} \equiv \{\mathcal{B}_i\}_{i \in \mathcal{I}}$ is an atlas for $\mathcal{B}$. Each chart $i$ admits a trivialization map $T_i: \mathcal{G}_i \rightarrow \mathcal{P}_i$, so that $(\mathcal{P}_i, \mathcal{B}_i, \mathcal{G}_i)$ has the structure of a trivial quantum principal bundle. Finally, for every $(i,j) \in \mathcal{I}^{(2)}$ there exists a quantum gauge transformation $\Gamma_{ij}: \mathcal{G}_{ij} \rightarrow \mathcal{B}_{ij}$ such that $T_i = \bigg(\Pi_{ij} \circ \Gamma_{ij}\bigg) \star T_j$ when restricted to `overlapping' elements, and $\Gamma_{ij}$ satisfy a triple overlap condition $\Gamma_{ij} \star \Gamma_{jk} = \Gamma_{ik}$ for all $(i,j), (j,k) \in \mathcal{I}^{(2)}$.    
\end{definition}

In words, a quantum orbifold corresponds to an object like a locally trivial quantum principal bundle without further assumptions on the indexing $i \in \mathcal{I}$ beyond the existence of a relation between charts we refer to as intersection. This resolves both complaints above and allows us to realize that the G-framed algebra introduced in previous work is simply the quantum analogue of an orbifold, in the case where all of the $\mathcal{P}_i$ are von Neumann and each $\mathcal{G}_i = \mathcal{L}(G_i)$ for some locally compact group $G_i$~\cite{ahmad2024quantumreferenceframestopdown}. Just as the configuration space of a gauge theory is an orbifold~\cite{vanrietvelde_switching_2021}, the quantization of its phase space is a quantum version of an orbifold. This is further evidence for the conjecture that the crossed product (which is a trivial quantum principal bundle) is the quantization of the (local) extended phase space which is a classical principal bundle. In defining the G-framed algebra, we have simply elevated this idea to a global statement. 

In our original work, the G-framed algebra was defined through a quotienting procedure where one identifies overlapping degrees of freedom in different QRFs. Indeed, this structure persists for the quantum orbifold since it may be the case that some QRFs share degrees of freedom signaling a semi-product structure on the index set. In the above notation, for $(i,j) \in \mathcal{I}^{(2)}$ the $i$-th and $j$-th QRFs contain an isomorphic subalgebra $\mathcal{P}_{ij}$ with base $\mathcal{B}_{ij}$. In that case, the consistency between the two different frames is implemented by the quantum gauge transformation $\Gamma_{ij}: \mathcal{G}_{ij} \rightarrow \mathcal{B}_{ij}$. As we have discussed above, this can be interpreted as specifying that the crossed products $\mathcal{P}_i$ and $\mathcal{P}_j$ are isomorphic when restricted to their intersection set $\mathcal{P}_{ij}$.

\section{Discussion}

One of the major themes of this work, and of the recent program of research produced by the authors \cite{ahmad2024quantumreferenceframestopdown,Klinger:2023qna,Jia:2023tki,Klinger:2023auu,Klinger:2023tgi,AliAhmad:2023etg,AliAhmad:2024eun}, is emphasizing the relationship between geometry and algebra. In this note, we have made this theme manifest by recasting our previous analyses of constrained physical systems entirely within the language of non-commutative geometry. From this point of view, the extended phase space is realized as a \emph{classical principal bundle}, the crossed product is realized as a \emph{trivial quantum principal bundle}, and the $G$-framed algebra is realized as a \emph{quantum orbifold}. While it is true that the non-commutative nature of the latter objects complicates many of the details concerning their geometry, contextualizing these algebraic structures by analogy to more familiar geometric spaces allows for the intuition of classical geometric analysis to be carried over into the quantum realm. 

In this note, we have exemplified the benefits of an approach which highlights the synergy between geometry and algebra in a variety of different ways. We first provided a rigorous definition of the extended phase space starting from Poisson geometry. This definition identifies the extended phase space as a Poisson manifold which can be functorially associated with a given Poisson covariant system, and as having the structure of a classical principal bundle. We then reviewed a conjecture that the extended phase space is a classical analog of the crossed product algebra~\cite{Klinger:2023auu,Klinger:2023tgi}. We substantiated this claim concretely by identifying the crossed product as a quantized principal bundle. The classical analysis of the extended phase space which we presented was based around a local chart and did not reckon with issues of global topology. In a similar vein, we proved that the crossed product is globally \textit{trivial} as quantum principal bundle.

Identifying crossed product algebras as trivial quantum principal bundles provides a new lens for unpacking the claim of \cite{ahmad2024quantumreferenceframestopdown} that a single crossed product algebra cannot admit multiple QRFs. This is supplemented by the characterization of quantum gauge transformations as inner conjugations of group automorphism actions. The global triviality of the crossed product implies that it admits only a single equivalence class of von Neumann covariant systems. When applied directly to the problem of constraint quantization, this allowed us to demonstrate that quantum gauge transformations implement an isomorphism at the level of gauge-invariant observables. 

Moving to the locally trivial quantum principal bundle, the global triviality of the crossed product identifies it as one chart amongst many making up a single, topologically non-trivial superstructure embedding all of the physical observables of a particular gauge invariant system. Moreover, we provided a novel definition of a quantum orbifold, which was motivated by clear physical and relational considerations. It is known that the configuration space of a constrained system is generically an orbifold since some constraints degenerate in local charts~\cite{1994CMaPh.160..431M}. The presence of inequivalent trivializations of the corresponding phase space necessitates inequivalent or orthogonal dressings for the same theory, which translate into the existence of inequivalent reference frames in the quantum context. The proposed definition of a quantum orbifold grounds our previous construction~\cite{ahmad2024quantumreferenceframestopdown} in the mature mathematical field of quantum geometry~\cite{Majid:1993re}. 

The non-commutative geometry of quantum orbifolds allows us to conceive of the algebraic structure underpinning a gauge theory in very much the same terms we might think of a geometric space equipped with an atlas of local charts. Change of frame maps have been considered in the QRF literature, and in this work we see that (1) changes in the \textit{internal} frame of a given dynamical system are simply quantum gauge transformations from the perspective of the quantum principal bundle and (2) the data of the quantum orbifold allows us to change QRFs on overlaps by a quantization of the classical notion of transition functions on a principal bundle. 

Hopefully, the preceding discussion has served to underscore how the geometric identification of the crossed product and its global generalizations motivates novel algebraic and physical insights. We see multiple avenues to explore inspired by the analysis presented in this note. We would like to list just a handful of these in the remainder of this section:

\begin{enumerate}
    \item It has been argued that the Gribov ambiguity of constrained systems amounts to inequivalent quantum reductions of the same theory~\cite{1994CMaPh.160..431M}. These can be interpreted as superselection sectors, and there is a rich classical cohomological theory underlying these ideas. As we have alluded to, the presence of non-trivial `topology' appears to signal the emergence of inequivalent QRFs. It is fascinating to ask whether this phenomenon might be understood in terms of `quantum cohomology'. It is possible to assign cohomology to a quantum principal bundle by appealing to a construction known as the universal differential calculus \cite{Brzezinski:1992hda,hajac1996strong}.
    \item Building upon the previous point, the universal differential calculus of quantum principal bundles leads to a natural identification of quantum connections and curvatures. The notion of triviality of the quantum bundles considered in this work points to the global factorization of the base and the fiber. This leaves behind the immediate question of what role the quantum analog of curvature plays in physical applications. It has been shown that quantum curvature carries more information than classical~\cite{Majid:1993re, durdevic1995characteristicclassesquantumprincipal, durdevic1996quantumprincipalbundlescharacteristic,G2003}. For this reason, it would be interesting to interpret invariants of the quantum curvature, like characteristic classes, physically and relationally. Since these invariants do not change under quantum gauge transformations, they have the interpretation of being frame independent objects.
    \item Finally, and reiterating a consideration originally addressed in \cite{ahmad2024quantumreferenceframestopdown}, it seems natural to point our new machinery towards quantum gravity away from the semiclassical limit~\cite{Witten:2023xze}. Roughly speaking, one can think of a semiclassical background as \emph{defining} a quantum reference frame. One might expect the full gravitational algebra to be a direct sum of all such backgrounds, but the resulting algebra would fail to be separable and the direct sum structure would suppress quantum overlaps between geometries. Taking the reference frame point of view seriously, however, the analysis of this paper tells us that to complete the gravitational algebra we must further construct change of frame maps which relate isomorphic observables across backgrounds. In this way, the quantum orbifold, with its non-trivial quotient structure, provides a useful algebraic construct for quantifying the overlap of non-diffeomorphic semiclassical backgrounds. It would be interesting to see how this can be explicitly realized in controlled settings where we can compare to path integral approaches.  A hint of this structure has appeared in the analysis of \cite{Ciambelli:2024swv}. There, it was argued that the appearance of a central charge in the quantization of gravity on null hypersurfaces facilitates the emergence of time. This is because the presence of the central charge forces a projective representation in which there is not a single unique ground state, but rather a full vacuum module. The states of the vacuum module were interpreted as quantum reference frames identifying distinct emergent times, but which nevertheless possess nontrivial overlaps except in the strict $G_N \rightarrow 0$ limit.
\end{enumerate}

\section*{Acknowledgment}
S.A.A. and M.S.K. would like to thank Ahmed Almheiri, Tom Faulkner, Samuel Goldman, Simon Lin,  Alexander Smith, and Michael Stone for interesting discussions.  R.G.L. would also like to thank the Perimeter Institute for Theoretical Physics for support, and acknowledges the existence of  U.S. Dept. of Energy grant DE-SC0015655.

\appendix
\renewcommand{\theequation}{\thesection.\arabic{equation}}
\setcounter{equation}{0}

 \section{Hopf-Algebras and Comodules} \label{app: Hopf}

\begin{definition}[Algebra]
	An algebra $A$ over a field $K$ is a vector space together with a pair of maps $m: A \otimes A \rightarrow A$ and $\eta: K \rightarrow A$ called multiplication and the unit such that multiplication is associative
	\beq
		m \circ (id_A \otimes m) = m \circ (m \otimes id_A),
	\eeq
	and the unit implements scalar multiplication
	\beq
		m \circ (\eta \otimes id_A)(z \otimes a) = za, \qquad m \circ (id_A \otimes \eta)(a \otimes z) = za, \; \forall z \in K, a \in A. 
	\eeq
	If the algebra $A$ is unital, we may regard $\eta(z) = z \mathbb{1}_A$. When it is clear or convenient to do so we may write $m(a_1 \otimes a_2) = a_1 a_2$.  
\end{definition}

\begin{definition}[Coalgebra]
	A coalgebra $C$ is a vector space over a field $K$ together with a pair of maps $\delta: C \rightarrow C \otimes C$ and $\epsilon: C \rightarrow K$ called comultiplication and the counit such that comultiplication is associative
	\beq
		(id_C \otimes \delta) \circ \delta = (\delta \otimes id_C) \circ \delta,
	\eeq
	and the counit is defined by
	\beq
		(\epsilon \otimes id_C) \circ \delta(c) = 1 \otimes c, \qquad (id_C \otimes \epsilon) \circ \delta(c) = c \otimes 1, \; \forall c \in C. 
	\eeq
	An element $c \in C$ is called grouplike if $\delta(c) = c \otimes c$ and $\epsilon(c) = 1$.\footnote{This is hopefully familiar from context of a group von Neumann algebra $\mathcal{L}(G)$ which admits a coaction $\delta_G: \mathcal{L}(G) \rightarrow \mathcal{L}(G)^{\otimes 2}$ such that $\delta_G(\ell(g)) = \ell(g) \otimes \ell(g)$ for all $g \in G$.} 
\end{definition}

One often uses so-called Sweedler notation \cite{sweedler-notation} to represent the comultiplication:
	\beq
		\delta(c) = \sum c_{(1)} \otimes c_{(2)}.
	\eeq
	This is meant to express the fact that generically $\delta(c)$ can be written as the sum of elements $c_{(1)}^i \otimes c_{(2)}^i$ over some index set $i \in \mathcal{I}$. For computations it is sufficient to keep track of the form of this sum without concerning one's self with the extent of the sum or the specific summands. We also define
	\beq
		\delta^n \equiv (\delta \otimes id_C) \circ \delta^{n-1} = (id_C \otimes \delta) \circ \delta^{n-1}: C \rightarrow C^{\otimes (n+1)},
	\eeq
	where the order of the composition is irrelevant by coassociativity. In Sweedler notation
	\beq
		\delta^n(c) = \sum \bigotimes_{j = 1}^{n+1} c_{(j)}. 
	\eeq
	At times it will add clarity to discussion if we slightly modify this notation to write
	\beq
		\delta^n(c) = \sum \bigotimes_{j = 1}^{n+1} \delta^n_{(j)}(c). 
	\eeq
	
\begin{definition}[Convolution]
	Let $(A,m,\eta)$ and $(C,\delta,\epsilon)$ be an algebra and a coalgebra, respectively over a common field $K$. These structures induce on the space of maps $\mathcal{L}(C,A)$ the structure of an algebra with a product called the convolution. In particular given $f_1, f_2: C \rightarrow A$ we define $f_1 \star f_2: C \rightarrow A$ by
	\beq
		(f_1 \star f_2)(c) \equiv m \circ (f_1 \otimes f_2) \circ \delta(c).
	\eeq
	In Sweedler notation:
	\beq
		(f_1 \star f_2)(c) = \sum m\bigg(f_1(c_{(1)}) \otimes f_2(c_{(2)})\bigg) = \sum f_1(c_{(1)}) f_2(c_{(2)}). 
	\eeq
	The map $\eta \circ \epsilon: C \rightarrow A$ acts an an identity for the algebra $\mathcal{L}(C,A)$ when endowed with the convolutional product:
	\beq
		f \star (\eta \circ \epsilon) = (\eta \circ \epsilon) \star f = f.
	\eeq
	We say that $f: C \rightarrow A$ is convolutionally invertible if there exists a map $f^{-1}: C \rightarrow A$ such that
	\beq
		f \star f^{-1} = f^{-1} \star f = \eta \circ \epsilon. 
	\eeq
\end{definition}
	
\begin{definition}[Bialgebra]
	A bialgebra $B$ is a vector space over a field $K$ which admits simultaneously the structure of an algebra $(B,m,\eta)$ and the structure of a coalgebra $(B,\delta,\epsilon)$ and for which either $(m,\eta)$ are algebra morphisms or $(\delta,\epsilon)$ are coalgebra morphisms. 
\end{definition}

\begin{definition}[Hopf Algebra]
	Let $(H,m,\eta,\delta,\epsilon)$ be a bialgebra. Since $H$ is both an algebra and a coalgebra, the space of linear maps $\mathcal{L}(H)$ can be endowed with a convolutional product. An endomorphism $S: H \rightarrow H$ is called an antipode if
	\beq
		id_H \star S = S \star id_H = \eta \circ \epsilon. 
	\eeq
	From the discussion above, we see that this identifies the antipode as the convolutional inverse of the identity map $id_H: H \rightarrow H$. A bialgebra equipped with an antipode is called a Hopf algebra.
\end{definition}

\begin{definition}[Comodules]
	Let $H$ be a Hopf algebra and $V$ a vector space. $V$ is called a left $H$-comodule if it admits a left coaction
	\beq
		L^V: V \rightarrow H \otimes V
	\eeq
	such that
	\beq
		(\delta \otimes id_V) \circ L^V = (id_H \otimes L^V) \circ L^V, \qquad (\epsilon \otimes id_V) \circ L^V = id_V. 
	\eeq
	If $V$ is a unital algebra and
	\beq
		L^V(vw) = L^V(v) L^V(w), \qquad L^V(\mathbb{1}_V) = \mathbb{1}_H \otimes \mathbb{1}_V,
	\eeq
	then we called $V$ a left $H$ comodule algebra. Similarly, if $V$ admits a right coaction
	\beq
		R^V: V \rightarrow V \otimes H
	\eeq 
	such that
	\beq \label{Right comodule}
		(R^V \otimes id_H) \circ R^V = (id_V \otimes \delta) \circ R^V, \qquad (id_V \otimes \epsilon) \circ R^V = id_V,
	\eeq
	then $V$ is called a right $H$-comodule. If $V$ is a unital algebra and $R^V$ is an algebra homomorphism, then $V$ is called a right $H$-comodule algebra. 
\end{definition}

Let $H$ be a Hopf algebra with a left $H$ comodule $V$. Moreover, let $B$ be an algebra with module $\Gamma$, e.g. we have a representation $\pi_B: B \rightarrow \text{End}(\Gamma)$ giving rise to an action $\Pi_B: B \otimes \Gamma \rightarrow \Gamma$ such that $\Pi^{(l)}_B(b \otimes \gamma) = \pi_B(b) \gamma$. Given this set up, we can extend the convolutional product to an action of $\mathcal{L}(H,B)$ on $\mathcal{L}(V,\Gamma)$. In particular, given $U: H \rightarrow B$ and $\psi: V \rightarrow \Gamma$ we define $U \star \psi \in \mathcal{L}(V,\Gamma)$ by
\beq
	U \star \psi(v) \equiv \Pi^{(l)}_B \circ (U \otimes \psi) \circ L^V(v) = \sum \Pi^{(l)}_B\bigg(U \circ L^V_{(1)}(v) \otimes \psi \circ L^V_{(2)}(v)\bigg). 
\eeq
Of course, in the case that $V$ is a right $H$ comodule and $B$ admits an antirepresentaton $\overline{\pi}_B: B \rightarrow \text{End}(\Gamma)$ lifting to a right action $\Pi_B^{(r)}: \Gamma \otimes B \rightarrow \Gamma$ we get a right convolutional action of $\mathcal{L}(H,B)$ on $\mathcal{L}(V,\Gamma)$ with
\beq
	\psi \star U(v) \equiv \Pi_B^{(r)} \circ (\psi \otimes U) \circ R^V(v) = \sum \Pi^{(r)}_{B}\bigg(\psi \circ R^V_{(1)}(v) \otimes U \circ R^{V}_{(2)}(v)\bigg). 
\eeq

\bibliographystyle{uiuchept}
\bibliography{RelationalQGeometry}
\end{document}